\documentclass[11pt]{article}
\usepackage[utf8]{inputenc}
\usepackage[left=1in, right=1in, bottom=1.15in, top=1.1in]{geometry}

\usepackage{tingfeng}

\title{Constant Inapproximability of Pacing Equilibria\\ in Second-Price Auctions}

\author{
Xi Chen\footnote{Supported by NSF grants IIS-1838154, CCF-2106429 and CCF-2107187.}\\ Columbia University\\\url{xichen@cs.columbia.edu}\hspace{-0.2cm}
\and Yuhao Li\footnote{Supported by NSF grants IIS-1838154, CCF-2106429 and CCF-2107187.}\\Columbia University\\\url{yuhaoli@cs.columbia.edu}\hspace{-0.2cm}
}

\date{}
\begin{document}
\maketitle
\begin{abstract}
In this paper, we revisit the problem of approximating a pacing equilibrium in second-price auctions, introduced by Conitzer, Kroer, Sodomka, and Moses [Oper. Res. 2022]. We~show~that finding a constant-factor approximation of a pacing equilibrium is $\PPAD$-hard, thereby strengthening previous results of Chen, Kroer, and Kumar [Math. Oper. Res. 2024], which established $\PPAD$-hardness only for inverse-polynomial approximations.
\end{abstract}

\section{Introduction}

Online advertisements have become a cornerstone of Internet economies, enabling advertisers to reach targeted audiences with remarkable precision. Auction-based mechanisms are at the heart of this ecosystem, determining the allocation and pricing of ad placements. Advertisers, on the other hand, who typically participate in thousands of auctions daily, must navigate budget constraints by adopting strategies that balance immediate bidding opportunities with long-term efficiency. 

One widely adopted method for budget management is \emph{pacing}, a strategy employed by major platforms like Google~\cite{google_ads_budget_pacing} to ensure smooth and efficient budget utilization over the course of an advertising campaign. By associating a pacing multiplier \hspace{0.05cm}---\hspace{0.05cm} a value between zero and one \hspace{0.05cm}---\hspace{0.05cm} with each advertiser, its bids are multiplicatively scaled down to prevent premature budget exhaustion while maximizing participation in later, potentially more valuable auctions. Besides its simplicity, pacing enjoys the property that allows a buyer to win items that provide the best return on investment (ratio of value to price) subject to their budget constraint.

In this paper, we examine pacing in second-price auctions, where bidding truthfully is a dominant strategy in the absence of budget constraints. However, in most cases, bidders face budget constraints, and they may adopt pacing strategies to manage their expenditures. In scenarios where all bidders employ pacing strategies, each bidder's optimal pacing multiplier is influenced by others' multipliers. The predicted outcomes by game theorists are often characterized as the stable points, or, equilibria, where every bidder's pacing multiplier is optimal given those of other bidders. This was formalized by Conitzer, Kroer, Sodomka, and Moses \cite{CKSM22} as a pacing equilibrium, where they showed such an equilibrium always exists in any second-price pacing game, whereas leaving the complexity of computing a pacing equilibrium as an open problem. 

In \cite{chen2024complexity},
Chen, Kroer, and Kumar investigated the computational complexity of pacing equilibria in second-price pacing games. On the positive side, they showed that computing an exact pacing equilibrium is in the complexity class $\PPAD$, thus implying that there always exists an exact rational pacing equilibrium. On the negative side, however, they showed that it is $\PPAD$-hard to compute an approximate pacing equilibrium when the approximation precision is inversely polynomial in the size of the market.

\paragraph{Our Contribution.} We revisit the approximability of pacing equilibria in second-price pacing games, and provide the first $\PPAD$-hardness for \emph{constant approximation}. Hence, a polynomial-time approximation scheme (PTAS) for pacing equilibria cannot exist unless $\PPAD = \FP$.

\begin{restatable}{theorem}{theoremmain}\label{theorem: main}
	For any constant $\gamma<1/3$, the problem of computing a $\gamma$-approximate pacing equilibrium in a second-price pacing game is $\PPAD$-hard. This holds even when each bidder bids on at most four items.
\end{restatable}

We give the definition of approximate pacing equilibria in \Cref{sec:preliminaries}, and
  prove \Cref{theorem: main} in \Cref{section: proof of main theorem}.
In fact, we also manage to establish $\PPAD$-hardness for a~weaker notion of approximate pacing equilibria defined in \cite{chen2024complexity}, which relaxes more constraints in the definition. We present its definition and the corresponding constant inapproximability result in \Cref{section: variant chen2024complexity}.

\subsection{Technical Overview}
For precise definitions of pacing equilibria and related concepts, we refer the reader to \Cref{sec:preliminaries}.

The insight of the whole reduction starts from a simple but crucial trick that can guarantee a lower bound for the pacing multipliers (that we care about) in any pacing equilibrium. To show the insight, suppose that there is a market where the budget of every buyer is at most 10. We are going to make a local change in this market to make sure that $\alpha_b\geq 1/2$, where $b$ is some buyer, in any pacing equilibrium.
This is achieved by introducing an auxiliary (rich) buyer $c$ and an auxiliary good $g$, where the budget of $c$ is a large constant (say 1000), and only $b$ and $c$ are interested in this good, and $c$ is only interested in the good $g$. We set $v_{b}(g)=200$, $v_{c}(g)=100$, and the budget of $b$ to be its original budget plus 100. By the ``not too much unnecessary pacing'' condition in the definition of approximate pacing equilibrium, it is easy to check that $\alpha_b\geq 1/2$ in any pacing equilibrium in this newly modified auction market. Moreover, there is a sharp distinction between the cases $\alpha_b=1/2$ and $\alpha_b>1/2$. If $\alpha_b> 1/2$, then the buyer $b$ must spend all of the additional 100 money on $g$, which means this trick will not raise the budget of the buyer $b$ when $\alpha_b> 1/2$, and it is possible that the bidder $b$ runs over of their budget. In contrast, when $\alpha_b=1/2$, the bidder $b$ ties with the auxiliary rich buyer $c$ on the good $g$, then the flexibility of pacing equilibrium allows the bidder $b$ to receive an appropriate amount of good $g$, making sure that their total spending always remains within budget.

We will show \Cref{theorem: main} by reducing from the $\purec$ problem, which was introduced recently by Deligkas, Fearnley, Hollender, and Melissourgos~\cite{DFHM24}. $\purec$ has been proven a remarkable technique for establishing strong constant inapproximability in $\PPAD$, such as Nash equilibrium in graphical games and polymatrix games~\cite{DFHM24}, Nash equilibrium in public goods games~\cite{DH24}, and competitive equilibrium in Fisher markets~\cite{DFHM24fisher}.  In a $\purec$ instance, we are given a set of nodes $[n]$ and a set of gates $C$. The goal is to find an assignment $\bx:[n]\mapsto\set{0,1,\bot}$ that satisfies all the constraints provided by the set of gates $C$. See the formal definitions of the gates in \Cref{section: pure circuit}.

In our reduction, for every variable $v$ in the $\purec$ instance, we will introduce a buyer, denoted by $b_v$, whose pacing multiplier $\alpha_{b_v}$ will encode that variable. Throughout the remaining technical overview, let us suppose that we can make sure $\alpha_{b_v}\geq 1/2$ in any pacing equilibrium for free using the aforementioned trick. As a result, the decoding works naturally as follows: $\bx[v]=1$ if $\alpha_{b_v}=1$; $\bx[v]=0$ if $\alpha_{b_v}=1/2$; and $\bx[v]=\bot$ otherwise. 

To illustrate the construction of the gates, take the $\NOT$ gate as an example. Suppose that $(u,v)$ is a $\NOT$ gate where $u$ is the input variable and $v$ is the output variable. We construct a good $g$ and set the value $v_{b_u}(g)=1$ and $v_{b_v}(g)=1000$ (this could be any number larger than $2\cdot v_{b_u}(g)$). For all other bidders, the value of $g$ is set to zero. Recall that we can make sure $\alpha_{b_v}\geq 1/2$, so we know that \emph{the buyer $b_v$ always wins the good $g$ and its payment on this auxiliary good is exactly $\alpha_{b_u}$} (since $v_{b_u}(g)$ is set to 1). Let the budget of $b_v$ to be $100+0.7$, where the 100 is from the trick above and $0.7$ is some number between $[0.5,1]$. Now we are ready to check the correctness: If $\alpha_{b_u}=0.5$, then there is at least $0.7-0.5=0.2$ amount of money not spent, which is roughly 0.19\% fraction of $b_v$'s budget, implying that $\alpha_{b_v}=1$ by the ``not too much unnecessary pacing'' condition when $\gamma<0.0019$; If $\alpha_{b_u}=1$, then $\alpha_{b_v}=1/2$ otherwise the budget of $b_v$ will not be enough for the payment.

To get stronger constant inapproximability, namely, 1/3 in the \Cref{theorem: main} above, we need to redesign all the numbers above and make sure that all parameters that we intend to use work well for all gates in the $\purec$ problem. The full detailed construction is presented in \Cref{section: proof of main theorem}.

Our reduction can be easily extended to prove $\PPAD$-hardness for constant approximation for a weaker variant of approximate pacing equilibrium defined in \cite{chen2024complexity}, which we show in \Cref{section: variant chen2024complexity}.

\subsection{Additional Related Works}
\paragraph{Auto-bidding and Auctions in Online Advertising.} Online advertising markets are increasingly adopting automated bidding strategies, reflecting a significant shift in the industry. This evolution has sparked a surge in research efforts across both academia and industry, as highlighted in the recent survey by Google~\cite{ABBBDFGLLMMMMLPP24}.

A particularly influential approach in this domain is the pacing strategy, which has been instrumental in developing online learning algorithms for various applications. Notable examples include learning under budget constraints~\cite{BalseiroGur19,BLM23,WYDK23}, budget and Return-on-Spend (RoS) constraints~\cite{FPW23,BB0LMS024}, and budget and Return-on-Investment (ROI) constraints~\cite{LPSZ24}. Additionally, several studies have explored scenarios where participants strategically misreport their private information to manipulate the pacing equilibrium for their benefit~\cite{FLS24,WLCRZD24}.

\paragraph{Equilibrium Computation.} 
The computation of equilibria for online auctions with constraints has attracted significant attention recently. 

For second-price auctions, Conitzer et al. \cite{CKSM22} introduced the notion of pacing equilibria and showed its existence. Chen et al. \cite{chen2024complexity} showed $\PPAD$-completeness for approximate pacing equilibria with an inversely-polynomial precision. Chen et al.~\cite{CKK21throttling} further studied the throttling strategy for budget management, defining the notion of throttling equilibria and showing its existence and $\PPAD$-completeness.
In sharp contrast, for first-price auctions, both pacing equilibria and throttling equilibria can be efficiently computed by simple t\^{a}tonnement-style dynamics~\cite{BCIJEM07,CKK21throttling} (see also \cite{conitzer2022pacing}). Recently, Li and Tang~\cite{ijcai2024p320} extended the pacing equilibrium to auto-bidding equilibrium to incorporate the return-on-spend (RoS) constraint and showed its existence and $\PPAD$-hardness for second-price auctions. 

Notably, for the $\PPAD$-hardness results associated with throttling equilibria and auto-bidding equilibria, constant inapproximability has been established~\cite{CKK21throttling,ijcai2024p320}.
Our work bridges the final gap in this domain by establishing the constant inapproximability of pacing equilibria.

\section{Preliminaries}\label{sec:preliminaries}

\subsection{Second-Price Pacing Games}

We start with the model of second-price pacing games. A second-price pacing game is an ordered tuple $G=(n,m,(v_{ij}),(B_i))$. There are $n$ buyers and $m$ (indivisible) goods, where $v_{ij}$ represents the value of good $j$ to the buyer $i$. Every buyer $i$ also has a budget, denoted by $B_i>0$. Each good is sold in an individual (single-slot) second-price auction. Every buyer participates in the auctions by using a pacing strategy, namely, the buyer $i$  chooses a pacing multiplier $\alpha_i\in[0,1]$ and bids $\alpha_i v_{ij}$ on every good $j$.

Each auction then operates similarly to a classic second-price auction. Let $h_j(\alpha)\coloneqq \max_{i}\alpha_i v_{ij}$ denote the highest bid on good $j$ and $p_j(\alpha)$ denote the second highest bid on good $j$. Notably, $h_j(\alpha)=p_j(\alpha)$ if there are two buyers with the same highest bids. Only buyers whose bids match $h_j(\alpha)$ are eligible to purchase a non-zero fraction of the good $j$ under the unit price $p_j(\alpha)$; in our setting with indivisible goods, the fraction of good $j$ received by buyer $i$
 should be interpreted as the probability of
  allocating good $j$ to buyer $i$.

Under mild conditions (for every $i$, there exists $j$ such that $v_{ij}>0$ and for every $j$, there exists $i$ such that $v_{ij}>0$), Conitzer et al. proved in \cite{CKSM22}  the following definition of pacing equilibria always exists in a second-price pacing game, where $x_{ij}$ in the vector $x=(x_{ij}) \in [0,1]^{nm}$ should be interpreted as the probability 
  of allocating good $j$ to buyer $i$:

\begin{definition}[Pacing Equilibria] \label{definition_exact_pacing_equilibrium}
Given a second-price pacing game $G = (n, m, (v_{ij}), (B_i))$, 
we say $(\alpha, x)$ with $\alpha=(\alpha_i)\in [0,1]^n$, $x=(x_{ij}) \in [0,1]^{nm}$
  and $\sum_{i\in [n]} x_{ij}\le 1$ for all $j\in [m]$ %
  is a \emph{pacing equilibrium} of $G$ if \vspace{0.15cm}
\begin{flushleft}
\begin{itemize}
    \item[(a)] Only buyers with the highest bid win the good: $x_{ij} > 0$ implies $\alpha_i v_{ij}= h_j(\alpha)$.\vspace{0.1cm}
    \item[(b)] Full allocation of each good with a positive bid:  $h_j(\alpha) > 0$  implies $\sum_{i\in [n]} x_{ij} = 1$.\vspace{0.1cm}%
    \item[(c)] Budgets are satisfied: $\sum_{j\in [m]} x_{ij} p_j(\alpha) \leq B_i$.\vspace{0.1cm}
	\item[(d)] %
	No unnecessary pacing: $\sum_{j\in [m]} x_{ij} p_j(\alpha) < B_i$ implies $\alpha_i=1$.\vspace{0.15cm}
\end{itemize}
\end{flushleft}
\end{definition}

Given a second-price pacing game $G$, we are interested in the computation of a $\gamma$-approximate pacing equilibrium, which relaxes the ``no unnecessary pacing'' condition. See also \Cref{definition_pacing_equilibrium_1} for an even weaker notion of approximation.
\begin{definition}[Approximate Pacing Equilibria] \label{definition_pacing_equilibrium}
Given a second-price pacing game $G = (n, m, (v_{ij}),$ $(B_i))$ and a parameter $\gamma \in [0,1)$, we say $(\alpha, x)$, with $\alpha=(\alpha_i)\in [0,1]^n$, $x=(x_{ij}) \in [0,1]^{nm}$
  and $\sum_{i\in [n]} x_{ij}\le 1$ for all $j\in [m]$, is
a $\gamma$-\emph{approximate pacing equilibrium} of $G$ if  \vspace{0.1 cm}
\begin{flushleft}
\begin{itemize}
    \item[(a)] Only buyers with the highest bid win the good: 
    $x_{ij} > 0$ implies $\alpha_i v_{ij} \geq h_j(\alpha)$.\vspace{0.1 cm}
    \item[(b)] Full allocation of each good with a positive bid: $h_j(\alpha) > 0$ implies 
      $\sum_{i\in [n]} x_{ij} = 1$.\vspace{0.1 cm}
    \item[(c)] Budgets are satisfied: $\sum_{j\in [m]} x_{ij} p_j(\alpha) \leq B_i$.\vspace{0.1 cm}
	\item[(d)] \emph{Not too much} unnecessary pacing:  $\sum_{j\in [m]} x_{ij} p_j(\alpha) < (1 - \gamma)B_i$ implies $\alpha_i = 1$.\vspace{0.15cm}
\end{itemize}
\end{flushleft}
\end{definition}

\subsection{The $\purec$ Problem}
\label{section: pure circuit}
We prove \Cref{theorem: main} by reducing from the $\purec$ problem \cite{DFHM24}. An instance of the $\purec$ problem is defined over a set of nodes $[n]=\{1,\ldots,n\}$ and a set of gates $C$. The goal is to find an assignment $\bx: [n]\mapsto \set{0,1,\bot}$ that satisfies all the constraints provided by the gates $C$. The constraint of each gate $(T,u,v,w)\in C$, where $u,v,w\in [n]$ are distinct nodes and $T\in\set{\NOT,\NOR,\PURIFY}$, is defined as follows:\footnote{Truth tables of $\NOT$, $\NOR$ and $\PURIFY$ can be found in \cite{DFHM24}.}
\begin{itemize}
	\item If $T=\NOT$, then $u$ is viewed as the input and $v$ as the output ($w$ is unused), and $\bx$ satisfies \[\bx[u]\in\set{0,1}\ \Longrightarrow \ \bx[v]=\neg\bx[u].\]
	\item If $T=\NOR$, then $u$ and $v$ are viewed as inputs and $w$ as the output, and $\bx$ satisfies 
\begin{align*}
\bx[u]=1\vee\bx[v]=1\ &\Longrightarrow\ \bx[w]=0 \quad \text{ and }\\
 \bx[u]=0\wedge\bx[v]=0\ &\Longrightarrow\ \bx[w]=1.
 \end{align*}
	\item If $T=\PURIFY$, then $u$ is viewed as the input and $v$ and $w$ as the outputs, and $\bx$ satisfies \begin{align*}
	\set{\bx[v],\bx[w]}\cap\set{0,1}\neq \emptyset \quad \text{ and } \\ \bx[u]\in\set{0,1}\ \Longrightarrow\  \bx[v]=\bx[w]=\bx[u].
	\end{align*}
\end{itemize}
An important note is that we require every node to be the output of exactly one gate. We also define the in-degree and out-degree of a node of a $\purec$ instance, based on its associated interaction graph. In this graph, there is an edge from node $u$ to node $v$ whenever $v$ is the output of a gate with input $u$. 

\begin{theorem}[\cite{DFHM24}]
	$\purec$ with gate set $\set{\emph{\NOT,\NOR,\PURIFY}}$ is $\PPAD$-complete. This holds even when the total degree of every node is at most 3; more specifically, for every node, the in- and out-degrees, $d_{in}$ and $d_{out}$, satisfy $(d_{in},d_{out})\in\set{(1,1),(2,1),(1,2)}$.
\end{theorem}

For convenience, we will use the $\NPURIFY$ gate instead of  the $\PURIFY$ gate in  $\purec$,
  where $\NPURIFY$ is just the negation of $\PURIFY$, namely 
\begin{itemize}
	\item If $T=\NPURIFY$, then $u$ is viewed as the input and $v$ and $w$ as the outputs, and $\bx$ satisfies \begin{align*}
	\set{\bx[v],\bx[w]}\cap\set{0,1}\neq \emptyset \quad \text{ and }\\ \bx[u]\in\set{0,1}\ \Longrightarrow \ \bx[v]=\bx[w]=\neg\bx[u].
	\end{align*}
\end{itemize}

It is clear that each $\PURIFY$ gate can be simulated by a $\NPURIFY$ gate combined with two NOT gates
  on its outputs. This gives a polynomial-time reduction from $\purec$ with gate set $\set{\NOT,\NOR,\PURIFY}$ to 
  $\purec$ with gate set $\set{\NOT,\NOR,\NPURIFY}$, and we get the following corollary:

\begin{corollary}
	$\purec$ with gate set $\set{\emph{\NOT,\NOR,\NPURIFY}}$ is $\PPAD$-complete. This holds even when the total degree of every node is at most 3.
\end{corollary}

\section{Proof of \Cref{theorem: main}}
\label{section: proof of main theorem}

Recall our main theorem: 
\theoremmain*

We give a polynomial-time reduction from $\purec$ to the problem of computing a $\gamma$-approximate pacing equilibrium. 
Let $\delta=1/3-\gamma>0$ and $\kappa=3\delta/2$. 
Given a pure circuit over nodes $[n]$ with $\set{ {\NOT,\NOR,\NPURIFY}}$-gates $C$,
  we will use $\delta$ and $\kappa$ as parameters in our construction of a second-price pacing game $G$ as follows.

\paragraph{Variable encodings.} For each variable $v\in [n]$ in the $\purec$ instance, we introduce a buyer, denoted by $b_v$, whose pacing multiplier $\alpha_{b_v}$ will encode the variable $\bx[v]$ of $v$. In particular, our construction will make sure that $\alpha_{b_v}\in[\kappa,1]$ for all $v\in[n]$ in any $\gamma$-approximate pacing equilibrium (\Cref{lemma: general range}). Given a vector of the pacing multipliers $\alpha$, we will extract an assignment for the $\purec$ instance as follows: 
\begin{equation}\label{eq:convert}
	\bx[v]\coloneqq
	\left\{ \begin{array}{ll}
	0 & \text{if $\alpha_{b_v}=\kappa$};\\
	1 & \text{if $\alpha_{b_v}=1$};\\
	\bot & \text{otherwise}.
	\end{array}
	\right.
\end{equation}
We next describe how to simulate each $\NOT$, $\NOR$ and $\NPURIFY$ gate in the pacing game, after which we prove the correctness
  of the reduction. 
Recall that every node $v$ is used as the output of exactly one gate in $C$.
As the construction below goes through each gate in $C$, we set the budget of the buyer that
  corresponds to the output of the gate; this way the budget of every buyer $b_v$ is well defined at the end.

\paragraph{$\NOT$ gates.}
Fix a $\NOT$ gate $(u,v,w,\NOT)$ in $C$, where $u$ is the input and $v$ is the output (and $w$ is unused).
We set the budget of $b_v$ to be 2.

To simulate this $\NOT$ gate with output $v$, we create an auxiliary buyer $c_v$ as well as a good $g_v$.
Only buyers $b_v$ and $c_v$ have non-zero values on the good $g_v$:
\begin{equation*}
	\text{value of $g_v$ to a buyer}\coloneqq
	\left\{ \begin{array}{ll}
	(1+\delta)/\kappa & \text{if the buyer is $b_v$};\\
	1+\delta & \text{if the buyer is $c_v$};\\
	0 & \text{for every other buyer (in the final game)}.
	\end{array}
	\right.
\end{equation*}
Note that in the construction above we intend to make sure that 
$$v_{b_v}(g_{v})=v_{c_v}(g_{v})\big/\kappa.$$ 
We set the budget of buyer $c_v$ to be $1000\gg 1+\delta$.

We then create another good $g_{(u,v)}$. The values of $g_{(u,v)}$ to the buyers are as follows.
\begin{equation*}
	\text{value of $g_{(u,v)}$ to a buyer}\coloneqq
	\left\{ \begin{array}{ll}
	1 & \text{if the buyer is $b_u$};\\
	(1/\kappa)+1 & \text{if the buyer is $b_v$};\\
	0 & \text{for every other buyer}.
	\end{array}
	\right.
\end{equation*}
The key property we need is  that 
$$v_{b_v}(g_{(u,v)})>v_{b_u}(g_{(u,v)})\big/\kappa.$$ Looking ahead, in \Cref{lemma: general range} we will show that $\alpha_{b_v}\geq \kappa$ in any $\gamma$-approximate pacing equilibrium. This will make sure that in any $\gamma$-approximate equilibrium, the buyer $b_v$ always wins the good $g_{(u,v)}$ and pays $\alpha_{b_u}$ for $g_{(u,v)}$ (because the bid from $b_u$ is $\alpha_{b_u}\cdot 1$ on $g_{(u,v)}$).

\paragraph{NOR gates.}

Fix a $\NOR$ gate $(u,v,w,\NOR)$ in $C$. 
We set the budget of $b_w$ to be $3$.

To simulate a $\NOR$ gate, we create an auxiliary buyer $c_w$ and a new good $g_w$. Only $b_w$ and $c_w$ have non-zero values on the good $g_w$. In particular, we let 
\begin{equation*}
\text{value of $g_{w}$ to a buyer}\coloneqq
	\left\{ \begin{array}{ll}
	(2-\delta)/\kappa & \text{if the buyer is $b_w$};\\
	2-\delta & \text{if the buyer is $c_w$};\\
	0 & \text{for every other buyer}.
	\end{array}
	\right.
\end{equation*}
Again, we intend to make sure that $$v_{b_w}(g_{w})=v_{c_w}(g_{w})\big/\kappa.$$ We set the budget of $c_w$ to be $1000\gg 2-\delta$.

We then create two new goods $g_{(u,w)}$ and $g_{(v,w)}$. The values of these goods to the buyers are 
\begin{align*}
	\text{value of $g_{(u,w)}$ to a buyer}&\coloneqq
	\left\{ \begin{array}{ll}
	1 & \text{if the buyer is $b_u$};\\
	1/\kappa+1 & \text{if the buyer is $b_w$};\\
	0 & \text{for every other buyer}.
	\end{array}
	\right.
	\quad \text{and}\\[1ex]
	\text{value of $g_{(v,w)}$ to a buyer}&\coloneqq
	\left\{ \begin{array}{ll}
	1 & \text{if the buyer is $b_v$};\\
	1/\kappa+1 & \text{if the buyer is $b_w$};\\
	0 & \text{for every other buyer}.
	\end{array}
	\right.
\end{align*}
The key property is that $$v_{b_w}(g_{(u,w)})>v_{b_u}(g_{(u,w)})\big/\kappa\quad\text{ and }\quad 
v_{b_w}(g_{(v,w)})>v_{b_v}(g_{(v,w)})\big/\kappa.$$ Looking ahead, in \Cref{lemma: general range} we will show that $\alpha_{b_w}\geq \kappa$ in any $\gamma$-approximate pacing equilibrium. This will make sure that in any $\gamma$-approximate equilibrium, the buyer $b_w$ always wins both goods $g_{(u,w)}$ and $g_{(v,w)}$ and pays $\alpha_{b_u}+\alpha_{b_v}$ on these two goods.

\paragraph{NPURIFY gates.}
Fix a $\NPURIFY$ gate $(u,v,w,\NPURIFY)$ in $C$. (Recall that both $v$ and $w$ are outputs.)
We set the budget of both $b_v$ and $b_w$ to be $3/2$. 

Next we create an auxiliary buyer $c_v$ and a new good $g_v$. Only $b_v$ and $c_v$ have non-zero values on $g_v$. Similarly, for the buyer $b_w$, we create an auxiliary buyer $c_w$ and a new good $g_w$, and only $b_w$ and $c_w$ have non-zero values on $g_w$. In particular, we set  
\begin{align*}
	\text{value of $g_{v}$ to a buyer}&\coloneqq
	\left\{ \begin{array}{ll}
	(1-\delta/2)/\kappa & \text{if the buyer is $b_v$};\\
	1-\delta/2 & \text{if the buyer is $c_v$};\\
	0 & \text{for every other buyer}.
	\end{array}
	\right.
	\quad \text{and}\\[1ex] \quad
	\text{value of $g_{w}$ to a buyer}&\coloneqq
	\left\{ \begin{array}{ll}
	(1/2+\delta/2)/\kappa & \text{if the buyer is $b_w$};\\
	1/2+\delta/2 & \text{if the buyer is $c_w$};\\
	0 & \text{for every other buyer}.
	\end{array}
	\right.
\end{align*}
Note that a key property in the construction above is that $$v_{b_v}(g_{v})=v_{c_v}(g_{v})\big/\kappa\quad\text{ and }\quad
v_{b_w}(g_{w})=v_{c_w}(g_{w})\big/\kappa.$$ We set the budget of $c_v$ to be $1000\gg 1-\delta/2$ and the budget of $c_w$ to be $1000\gg 1/2+\delta/2$.

We then create two new goods: $g_{(u,v)}$ and $g_{(u,w)}$. The values of these two goods are
\begin{align*}
	\text{value of $g_{(u,v)}$ to a buyer}&\coloneqq
	\left\{ \begin{array}{ll}
	1 & \text{if the buyer is $b_u$};\\
	1/\kappa+1 & \text{if the buyer is $b_v$};\\
	0 & \text{for every other buyer}.
	\end{array}
	\right.
	\quad \text{and}\\[1ex] \quad
	\text{value of $g_{(u,w)}$ to a buyer}&\coloneqq
	\left\{ \begin{array}{ll}
	1 & \text{if the buyer is $b_u$};\\
	1/\kappa+1 & \text{if the buyer is $b_w$};\\
	0 & \text{for every other buyer}.
	\end{array}
	\right.
\end{align*}
The key properties of the construction above are that $$v_{b_v}(g_{(u,v)})>v_{b_u}(g_{(u,v)})\big/\kappa\quad \text{ and }\quad v_{b_w}(g_{(u,w)})>v_{b_u}(g_{(u,w)})\big/\kappa.$$ Again, looking ahead, in \Cref{lemma: general range} we will show that $\alpha_{b_v}\geq \kappa$ and $\alpha_{b_w}\geq \kappa$ in any $\gamma$-approximate pacing equilibrium. This makes sure that in any $\gamma$-approximate equilibrium, the buyer $b_v$ always wins $g_{(u,v)}$ and pays $\alpha_{b_u}$ for $g_{(u,v)}$; the buyer $b_w$ always wins  $g_{(u,w)}$ and pays $\alpha_{b_u}$ for $g_{(u,w)}$.

\subsection{Proof of Correctness}

This finishes the construction of the pacing game $G$, which clearly can be done in polynomial time.
To prove the correctness of the reduction, we show in the rest of this section that given any $\alpha$ in
  a $\gamma$-approximate pacing equilibrium, $\bx$ obtained from $\alpha$ using \Cref{eq:convert}
  must be a solution to the given $\purec$ instance.
  
We start by showing that $\alpha$ must satisfy $\alpha_{b_v}\in[\kappa,1]$ for all $v\in[n]$. Note that this in particular implies that there is a strictly positive bid on every item; thus, the buyers are not allocated any item for which they have zero value.

\begin{lemma}\label{lemma: general range}
	Let $\alpha$ be the multipliers in any $\gamma$-approximate pacing equilibrium of $G$. Then we must have $\alpha_{c_v}=1$ and $\alpha_{b_v}\in[\kappa,1]$ for all $v\in[n]$.
\end{lemma}
\begin{proof}
	We first note that in the construction, $c_v$'s budget is so large that by the no unnecessary pacing condition, we have $\alpha_{c_v}=1$ for all $v\in[n]$. Now we show that $\alpha_{b_v}\geq \kappa$ for all $v\in[n]$. At a high level, the reasoning is exactly the same for the three different gates, namely, if $\alpha_{b_v}<\kappa$, then $b_v$ would not be able to buy $g_v$, which (we will show) implies that $b_v$ will spend less than $(1-\gamma)$ fraction of its budget, leading to a contradiction.
\begin{flushleft}\begin{enumerate}	
\item	\textbf{$\NOT$ gates.} Consider a $\NOT$ gate $(u,v,w,\NOT)$ and assume for a contradiction that $\alpha_{b_v}<\kappa$. Then we know that $\alpha_{b_v}\cdot v_{b_v}(g_v)<\alpha_{c_v}\cdot v_{c_v}(g_v)$, which means that $b_v$ cannot spend any money on the good $g_v$. Then, $b_v$'s overall expense is at most $v_{b_u}(g_{(u,v)})=1$ (since it is a second-price auction). Note that the budget of $b_v$ is 2. Thus, $b_v$'s expense is at most $1/2<(1-\gamma)$-fraction of its budget, leading to a contradiction.
	
\item	\textbf{$\NOR$ gates.} Consider a $\NOR$ gate $(u,v,w,\NOR)$ and assume for a contradiction that $\alpha_{b_w}<\kappa$. Then we know that $\alpha_{b_w}\cdot v_{b_w}(g_w)<\alpha_{c_w}\cdot v_{c_w}(g_w)$, which means that $b_w$ cannot spend any money on the good $g_w$. Then, $b_w$'s overall expense is at most $$v_{b_u}(g_{(u,w)})+v_{b_v}(g_{(v,w)})=2$$ (since it is a second price auction). Note that the budget of $b_w$ is $3$. Thus, $b_w$'s expense is at most $2/3<(1-\gamma)$-fraction of its budget, leading to a contradiction.
	
\item	\textbf{$\NPURIFY$ gates.} Consider a $\NPURIFY$ gate $(u,v,w,\NPURIFY)$ and  
	assume for a contradiction that $\alpha_{b_v}<\kappa$ or $\alpha_{b_w}<\kappa$. 
When $\alpha_{b_v}<\kappa$, we have $\alpha_{b_v}\cdot v_{b_v}(g_v)<\alpha_{c_v}\cdot v_{c_v}(g_v)$, which means that $b_v$ cannot spend any money on the good $g_v$. Then, $b_v$'s overall expense is at most $v_{b_u}(g_{(u,v)})=1$. Note that the budget of $b_v$ is $3/2$. Thus, $b_v$'s expense is at most $2/3<(1-\gamma)$-fraction of its budget, leading to a contradiction.
	
	When $\alpha_{b_w}<\kappa$, we have $\alpha_{b_w}\cdot v_{b_w}(g_w)<\alpha_{c_w}\cdot v_{c_w}(g_w)$, which means that $b_w$ cannot spend any money on the good $g_w$. Then, $b_w$'s overall expense is at most $v_{b_u}(g_{(u,w)})=1$. Note that the budget of $b_w$ is $3/2$. Thus, $b_w$'s expense is at most $2/3<(1-\gamma)$-fraction of its budget, leading to a contradiction.
\end{enumerate}\end{flushleft}	 
	 This finishes the proof of the lemma.
\end{proof}
\begin{lemma}\label{lemma: potential winning goods}
	Let $\alpha$ be the multipliers of any $\gamma$-approximate pacing equilibrium of $G$. Then for any good $g_{(u,v)}$, the buyer $b_v$ buys the whole unit of $g_{(u,v)}$ and pays $\alpha_{b_u}$ on $g_{(u,v)}$.
\end{lemma}
\begin{proof}
	By \Cref{lemma: general range}, we know that $\alpha_{b_v}\geq \kappa$. Then we have 
	$$\alpha_{b_v}\cdot v_{b_v}(g_{(u,v)})> v_{b_u}(g_{(u,v)})\geq \alpha_{b_u}\cdot v_{b_u}(g_{(u,v)}).$$ Since only buyers with the highest bid win the goods and pay the second price, we conclude that $b_v$ buys the whole unit of $g_{(u,v)}$ and pays $\alpha_{b_u}$ for $g_{(u,v)}$  (recall that in the construction, $v_{b_u}(g_{(u,v)})=1$).
\end{proof}

We are now ready to prove three lemmas, one for the correctness of each type of gate:

\begin{lemma}[Correctness of $\NOT$ gates]
	Let $\alpha$ be the multipliers of any $\gamma$-approximate pacing equilibrium of $G$. For any gate $(u,v,w,\NOT)$, if $\alpha_{b_u}=\kappa$, then $\alpha_{b_v}=1$; if $\alpha_{b_u}=1$, then $\alpha_{b_v}=\kappa$.
\end{lemma}
\begin{proof}
	By \Cref{lemma: potential winning goods}, we know that $b_v$ must buy the whole unit of $g_{(u,v)}$, may buy some amount of $g_v$, and does not buy any amount of other goods. So $b_v$'s overall expense is at most $\alpha_{b_u}+v_{c_v}(g_v)$. Furthermore, if $\alpha_{b_v}>\kappa$, then $b_v$  wins the whole unit of  $g_v$ and its overall  expense equals $\alpha_{b_u}+v_{c_v}(g_v)$. 
	
	If $\alpha_{b_u}=\kappa$, then $b_v$'s expense is at most $\kappa+(1+\delta)$. Recall that $b_v$'s budget is 2. So $b_v$ spends at most $(1+\delta+\kappa)/2<(1-\gamma)$-fraction of its budget. Thus by the no unnecessary pacing condition, we have that $\alpha_{b_v}=1$.
	
	For the case when $\alpha_{b_u}=1$, assume for contradiction that $\alpha_{b_v}>\kappa$. Then $b_v$'s expense equals $\alpha_{b_u}+v_{c_v}(g_v)=2+\delta$, violating the budget constraint and leading to a contradiction. Combining with \Cref{lemma: general range} which shows $\alpha_{b_v}\geq \kappa$, we conclude that if $\alpha_{b_u}=1$, then $\alpha_{b_v}=\kappa$.
\end{proof}

\begin{lemma}[Correctness of $\NOR$ gates]
	Let $\alpha$ be the multipliers of any $\gamma$-approximate pacing equilibrium of $G$. For any gate $(u,v,w,\NOR)$, if $\alpha_{b_u}=\kappa$ and $\alpha_{b_v}=\kappa$, then $\alpha_{b_w}=1$; if $\alpha_{b_u}=1$ or $\alpha_{b_v}=1$, then $\alpha_{b_w}=\kappa$.
\end{lemma}
\begin{proof}
	By \Cref{lemma: potential winning goods}, we know that $b_w$ must buy the whole unit of $g_{(u,w)}$ and the whole unit of $g_{(v,w)}$, may buy some amount of $g_w$, and does not buy any amount of other goods. So $b_v$'s overall expense is at most $\alpha_{b_u}+\alpha_{b_v}+v_{c_w}(g_w)$. Furthermore, if $\alpha_{b_w}>\kappa$, then $b_w$ also wins the whole unit of  $g_w$, which implies that $b_w$'s expense equals $\alpha_{b_u}+\alpha_{b_v}+v_{c_w}(g_w)$.
	
	If $\alpha_{b_u}=\kappa$ and $\alpha_{b_v}=\kappa$, then $b_v$'s expense is at most $\kappa+\kappa+(2-\delta)$. Recall that $b_w$'s budget is 3. So $b_w$ spends at most $(2+2\kappa-\delta)/3<(1-\gamma)$-fraction of its budget. Thus by the no unnecessary pacing condition, we have that $\alpha_{b_w}=1$.
	
	For the case when $\alpha_{b_u}=1$ or $\alpha_{b_v}=1$, assume for a contradiction that $\alpha_{b_w}>\kappa$. Then $b_w$'s expense equals $\alpha_{b_u}+\alpha_{b_v}+v_{c_w}(g_w)\geq 1+\kappa+2-\delta=3+\kappa-\delta>3$, violating the budget constraint and leading to a contradiction. Combining with \Cref{lemma: general range} which shows $\alpha_{b_w}\geq \kappa$, we conclude that if $\alpha_{b_u}=1$ or $\alpha_{b_v}=1$, then $\alpha_{b_w}=\kappa$.
\end{proof}

\begin{lemma}[Correctness of $\NPURIFY$ gates]
	Let $\alpha$ be the multipliers of any $\gamma$-approximate pacing equilibrium of $G$. For any gate $(u,v,w,\NPURIFY)$, we have $\alpha_{b_v}=\kappa$ or $\alpha_{b_w}=1$. Furthermore, if $\alpha_{b_u}=\kappa$, then $\alpha_{b_v}=\alpha_{b_w}=1$ and if $\alpha_{b_u}=1$, then $\alpha_{b_v}=\alpha_{b_w}=\kappa$.
\end{lemma}
\begin{proof}
	By \Cref{lemma: potential winning goods}, we know that $b_v$ must buy the whole unit of $g_{(u,v)}$, may buy some amount of $g_v$, and does not buy any amount of other goods. So $b_v$'s overall expense is at most $\alpha_{b_u}+v_{c_v}(g_v)$. Furthermore, if $\alpha_{b_v}>\kappa$, then $b_v$ also wins the whole unit of $g_v$ and its expense equals $\alpha_{b_u}+v_{c_v}(g_v)$.	
	
	Similarly, we know that $b_w$ must buy the whole unit of $g_{(u,w)}$, may buy some amount of $g_w$, and does not buy any amount of other goods. So $b_w$'s overall expense is at most $\alpha_{b_u}+v_{c_w}(g_w)$. Furthermore, if $\alpha_{b_w}>\kappa$, then $b_w$ also wins the whole unit of $g_w$, which implies that $b_w$'s expense equals $\alpha_{b_u}+v_{c_w}(g_w)$.	
	
	We will show the following cases, given which the lemma directly follows.
	\begin{itemize}
		\item When $\alpha_{b_u}=\kappa$, we have $\alpha_{b_v}=1$;
		\item When $\alpha_{b_u}\in(1/2+\delta/2,1]$, we have $\alpha_{b_v}=\kappa$;
		\item When $\alpha_{b_u}\in[\kappa,1/2+\delta/2]$, we have $\alpha_{b_w}=1$;
		\item When $\alpha_{b_u}=1$, we have $\alpha_{b_w}=\kappa$.
	\end{itemize}
	
	When $\alpha_{b_u}=\kappa$, $b_v$'s expense is at most $\kappa+(1-\delta/2)$. Recall that $b_v$'s budget is $3/2$. So $b_v$ spends at most $2(\kappa+1-\delta/2)/3<(1-\gamma)$-fraction of its budget. So by the no unnecessary pacing condition, we have that $\alpha_{b_v}=1$.

	When $\alpha_{b_u}>1/2+\delta/2$, assume for a contradiction that $\alpha_{b_v}>\kappa$. Then $b_v$'s expense equals $\alpha_{b_u}+v_{c_v}(g_v)>1/2+\delta/2+1-\delta/2=3/2$, violating the budget constraint. Combining with \Cref{lemma: general range}  we conclude that if $\alpha_{b_u}>1/2+\delta/2$, then $\alpha_{b_v}=\kappa$.
	
	When $\alpha_{b_u}\leq 1/2+\delta/2$, $b_w$'s expense is at most $1/2+\delta/2+1/2+\delta/2=1+\delta$. Recall that $b_w$'s budget is $3/2$. So $b_w$ spends at most $2(1+\delta)/3<(1-\gamma)$-fraction of its budget. Thus by the no unnecessary pacing condition, we have that $\alpha_{b_w}=1$.
	
	When $\alpha_{b_u}=1$, we assume for a contradiction that $\alpha_{b_w}>\kappa$. Then $b_w$'s overall expense equals $\alpha_{b_u}+v_{c_w}(g_w)=1+1/2+\delta/2=3/2+\delta/2$, violating the budget constraint. Again combining with \Cref{lemma: general range}, we conclude that if $\alpha_{b_u}=1$, then $\alpha_{b_w}=\kappa$.
\end{proof}
Combining  all three lemmas above, it is clear that given the multipliers $\alpha$ of any $\gamma$-approximate pacing equilibrium of $G$, we can extract a solution to the given $\purec$ instance in polynomial time. This finishes the proof of \Cref{theorem: main}.

\section{Discussion and Open Questions}
We showed constant inapproximability of pacing equilibrium in second-price auctions, ruling out PTAS for this problem unless $\PPAD=\FP$. However, the tight approximabilities of equilibria in auto-bidding settings remain largely open, and they are worth further investigation. We list a few concrete directions below:

\begin{itemize}
	\item Is there a polynomial-time algorithm with a non-trivial approximation guarantee, for example, one that computes a 0.99-approximate pacing equilibrium for second-price auctions?
	\item Is the 1/3 bound tight for our constructions? We do not know whether there is a polynomial-time algorithm that can find a 1/3-approximate pacing equilibrium for the family of instances we constructed. Slightly more broadly, is there a 1/3-approximate algorithm when each bidder bids on at most four items?
	\item It is known that when $n$ is a constant, there is a polynomial-time algorithm that can find an exact pacing equilibrium~\cite{yangliu2025}. What about the case where $m$ is a constant?
	\item What is the tight approximability of throttling equilibria in second-price auctions? 
	\item It is known that an exact pacing equilibrium for first-price auctions can be computed in time $\poly(n,m)$~\cite{conitzer2022pacing}. Is there an algorithm that computes $\delta$-approximate throttling equilibrium for first-price auctions in time $\poly(n,m,\log(1/\delta))$, improving on the previous $\poly(n,m,1/\delta)$ algorithm~\cite{CKK21throttling}?
\end{itemize}

\begin{flushleft}
\bibliographystyle{alpha}
\bibliography{ref}

\newcommand{\etalchar}[1]{$^{#1}$}
\begin{thebibliography}{DFHM24b}

\bibitem[ABB{\etalchar{+}}24]{ABBBDFGLLMMMMLPP24}
Gagan Aggarwal, Ashwinkumar Badanidiyuru, Santiago~R. Balseiro, Kshipra Bhawalkar, Yuan Deng, Zhe Feng, Gagan Goel, Christopher Liaw, Haihao Lu, Mohammad Mahdian, Jieming Mao, Aranyak Mehta, Vahab Mirrokni, Renato~Paes Leme, Andr{\'{e}}s Perlroth, Georgios Piliouras, Jon Schneider, Ariel Schvartzman, Balasubramanian Sivan, Kelly Spendlove, Yifeng Teng, Di~Wang, Hanrui Zhang, Mingfei Zhao, Wennan Zhu, and Song Zuo.
\newblock Auto-bidding and auctions in online advertising: {A} survey.
\newblock {\em SIGecom Exch.}, 22(1):159--183, 2024.

\bibitem[BBF{\etalchar{+}}24]{BB0LMS024}
Santiago~R. Balseiro, Kshipra Bhawalkar, Zhe Feng, Haihao Lu, Vahab Mirrokni, Balasubramanian Sivan, and Di~Wang.
\newblock A field guide for pacing budget and {ROS} constraints.
\newblock In {\em {ICML}}. OpenReview.net, 2024.

\bibitem[BCI{\etalchar{+}}07]{BCIJEM07}
Christian Borgs, Jennifer~T. Chayes, Nicole Immorlica, Kamal Jain, Omid Etesami, and Mohammad Mahdian.
\newblock Dynamics of bid optimization in online advertisement auctions.
\newblock In {\em {WWW}}, pages 531--540. {ACM}, 2007.

\bibitem[BG19]{BalseiroGur19}
Santiago~R. Balseiro and Yonatan Gur.
\newblock Learning in repeated auctions with budgets: Regret minimization and equilibrium.
\newblock {\em Manag. Sci.}, 65(9):3952--3968, 2019.

\bibitem[BLM23]{BLM23}
Santiago~R. Balseiro, Haihao Lu, and Vahab Mirrokni.
\newblock The best of many worlds: Dual mirror descent for online allocation problems.
\newblock {\em Oper. Res.}, 71(1):101--119, 2023.

\bibitem[CKK21]{CKK21throttling}
Xi~Chen, Christian Kroer, and Rachitesh Kumar.
\newblock Throttling equilibria in auction markets.
\newblock In {\em {WINE}}, volume 13112 of {\em Lecture Notes in Computer Science}, page 551. Springer, 2021.

\bibitem[CKK24]{chen2024complexity}
Xi~Chen, Christian Kroer, and Rachitesh Kumar.
\newblock The complexity of pacing for second-price auctions.
\newblock {\em Mathematics of Operations Research}, 49(4):2109--2135, 2024.

\bibitem[CKP{\etalchar{+}}22]{conitzer2022pacing}
Vincent Conitzer, Christian Kroer, Debmalya Panigrahi, Okke Schrijvers, Nicolas~E Stier-Moses, Eric Sodomka, and Christopher~A Wilkens.
\newblock Pacing equilibrium in first price auction markets.
\newblock {\em Management Science}, 68(12):8515--8535, 2022.

\bibitem[CKSM22]{CKSM22}
Vincent Conitzer, Christian Kroer, Eric Sodomka, and Nicol{\'{a}}s E.~Stier Moses.
\newblock Multiplicative pacing equilibria in auction markets.
\newblock {\em Oper. Res.}, 70(2):963--989, 2022.

\bibitem[CYW{\etalchar{+}}24]{WLCRZD24}
Zhaohua Chen, Mingwei Yang, Chang Wang, Jicheng Li, Zheng Cai, Yukun Ren, Zhihua Zhu, and Xiaotie Deng.
\newblock Budget-constrained auctions with unassured priors: Strategic equivalence and structural properties.
\newblock In {\em {WWW}}, pages 14--24. {ACM}, 2024.

\bibitem[DFHM24a]{DFHM24fisher}
Argyrios Deligkas, John Fearnley, Alexandros Hollender, and Themistoklis Melissourgos.
\newblock Constant inapproximability for fisher markets.
\newblock In {\em {EC}}, pages 13--39. {ACM}, 2024.

\bibitem[DFHM24b]{DFHM24}
Argyrios Deligkas, John Fearnley, Alexandros Hollender, and Themistoklis Melissourgos.
\newblock Pure-circuit: Tight inapproximability for {PPAD}.
\newblock {\em J. {ACM}}, 71(5):31:1--31:48, 2024.

\bibitem[DH24]{DH24}
J{\'{e}}r{\'{e}}mi~Do Dinh and Alexandros Hollender.
\newblock Tight inapproximability of nash equilibria in public goods games.
\newblock {\em Inf. Process. Lett.}, 186:106486, 2024.

\bibitem[FLS24]{FLS24}
Yiding Feng, Brendan Lucier, and Aleksandrs Slivkins.
\newblock Strategic budget selection in a competitive autobidding world.
\newblock In {\em {STOC}}, pages 213--224. {ACM}, 2024.

\bibitem[FPW23]{FPW23}
Zhe Feng, Swati Padmanabhan, and Di~Wang.
\newblock Online bidding algorithms for return-on-spend constrained advertisers.
\newblock In {\em {WWW}}, pages 3550--3560. {ACM}, 2023.

\bibitem[{Goo}25]{google_ads_budget_pacing}
{Google Ads}.
\newblock \url{https://ads.google.com/home/cost-tool/}, 2025.
\newblock Accessed: 2025-01-01.

\bibitem[LPSZ24]{LPSZ24}
Brendan Lucier, Sarath Pattathil, Aleksandrs Slivkins, and Mengxiao Zhang.
\newblock Autobidders with budget and {ROI} constraints: Efficiency, regret, and pacing dynamics.
\newblock In {\em {COLT}}, volume 247 of {\em Proceedings of Machine Learning Research}, pages 3642--3643. {PMLR}, 2024.

\bibitem[LT24]{ijcai2024p320}
Juncheng Li and Pingzhong Tang.
\newblock Vulnerabilities of single-round incentive compatibility in auto-bidding: Theory and evidence from roi-constrained online advertising markets.
\newblock In {\em {IJCAI}}, pages 2886--2894. ijcai.org, 2024.

\bibitem[WYDK23]{WYDK23}
Qian Wang, Zongjun Yang, Xiaotie Deng, and Yuqing Kong.
\newblock Learning to bid in repeated first-price auctions with budgets.
\newblock In {\em {ICML}}, volume 202 of {\em Proceedings of Machine Learning Research}, pages 36494--36513. {PMLR}, 2023.

\bibitem[YWL26]{yangliu2025}
Yonglei Yan, Zihe Wang, and Zhengyang Liu.
\newblock Pacing equilibria in second-price auctions with few buyers.
\newblock In {\em Proceedings of the AAAI Conference on Artificial Intelligence (AAAI)}, 2026.
\newblock To appear.

\end{thebibliography}
\end{flushleft}
\appendix
\section{Hardness of a weaker approximation}
\label{section: variant chen2024complexity}

We recall the following weaker notion of approximate pacing equilibria defined in \cite{chen2024complexity}:

\begin{definition}[Approximate Pacing Equilibria] \label{definition_pacing_equilibrium_1}
Given a second-price pacing game $G = (n, m, (v_{ij}),(B_i))$ and parameters $\sigma, \gamma, \tau \in [0,1)$, we say $(\alpha, x)$, with $\alpha=(\alpha_i)\in [0,1]^n$, $x=(x_{ij}) \in [0,1]^{nm}$
  and $\sum_{i\in [n]} x_{ij}\le 1$ for all $j\in [m]$, is
a $(\sigma,\gamma,\tau)$-\emph{approximate pacing equilibrium} of $G$ if  \vspace{0.1 cm}
\begin{flushleft}
\begin{itemize}
    \item[(a)] Only buyers \emph{close} to the highest bid win the good: %
    $x_{ij} > 0$ implies $\alpha_i v_{ij} \geq (1 - \sigma) h_j(\alpha)$.\vspace{0.1 cm}
    \item[(b)] Full allocation of each good with a positive bid: $h_j(\alpha) > 0$ implies 
      $\sum_{i\in [n]} x_{ij} = 1$.\vspace{0.1 cm}
    \item[(c)] Budgets are satisfied: $\sum_{j\in [m]} x_{ij} p_j(\alpha) \leq B_i$.\vspace{0.1 cm}
	\item[(d)] \emph{Not too much} unnecessary pacing:  $\sum_{j\in [m]} x_{ij} p_j(\alpha) < (1 - \gamma)B_i$ implies $\alpha_i \geq 1 - \tau$.\vspace{0.15cm}
\end{itemize}
\end{flushleft}
\end{definition}

Our main result in this section is constant inapproximability for this even more related approximate pacing equilibrium.

\begin{theorem}\label{thm: approx in CKK}
	For any constants $\sigma\leq 1/20,\gamma\leq 1/20,\tau \leq 1/20$, computing a $(\sigma,\gamma,\tau)$-approximate pacing equilibrium in second-price pacing games is $\PPAD$-hard. This holds even when every buyer has a non-zero value on at most four goods.
\end{theorem}
We note that the purpose for \Cref{thm: approx in CKK} is to show constant inapproximability for an even more relaxed definition of approximation studied in the literature. The constants in \Cref{thm: approx in CKK} are crude, and we did not manage to optimize them. Nevertheless, it is sufficient to rule out PTAS for approximate pacing equilibrium, assuming $\PPAD\neq \FP$.

\subsection{Proof of \Cref{thm: approx in CKK}}
\paragraph{Variable encodings.} For every variable $v$ in the $\purec$ instance, we will introduce a buyer, denoted by $b_v$, whose pacing multiplier $\alpha_{b_v}$ will encode that variable. In particular, our construction will make sure that $\alpha_{b_v}\geq 0.1$ for all $v\in[n]$ in any $(1/20,1/20,1/20)$-approximate pacing equilibrium (\Cref{lemma: general range ckk}). Given a vector of the pacing multipliers $\alpha$, we will extract an assignment for the $\purec$ instance as follows: 
\begin{equation*}
	\bx[v]\coloneqq
	\left\{ \begin{array}{ll}
	0 & \text{if $\alpha_{b_v}\in[0.1,0.15]$};\\
	1 & \text{if $\alpha_{b_v}\in [0.9,1]$};\\
	\bot & \text{otherwise}.
	\end{array}
	\right.
\end{equation*}

We construct the $\NOT$, $\NOR$, and $\NPURIFY$ gates below, after which we will prove their correctness. 

\paragraph{$\NOT$ gates.}
Fix a $\NOT$ gate $(u,v,w,\NOT)$ (where $w$ is unused) in $C$. 
We set the budget of $b_v$ to be 1.5.

For the buyer $b_v$, we create an auxiliary buyer $c_v$ and create a good $g_v$ for them, namely, only $b_v$ and $c_v$ have non-zero values on the good $g_v$. In particular, we let
\begin{equation*}
	v_{i}(g_{v})\coloneqq
	\left\{ \begin{array}{ll}
	9 & \text{if $i=b_v$};\\
	1 & \text{if $i=c_v$};\\
	0 & \text{otherwise}.
	\end{array}
	\right.
\end{equation*}
Note that in the construction above we intend to make sure $0.1\cdot v_{b_v}(g_{v})\leq (1-\sigma)v_{c_v}(g_{v})$, where $0.1$ is the intended lower bound for $\alpha_{b_v}$. We set the budget of $c_v$ to be $1000$. (Throughout this reduction, 1000 represents a number that is large enough.)

We then construct a good $g_{(u,v)}$. The values of $g_{(u,v)}$ to the buyers are as follows.

\begin{equation*}
	v_{i}(g_{(u,v)})\coloneqq
	\left\{ \begin{array}{ll}
	1 & \text{if $i=b_u$};\\
	1000 & \text{if $i=b_v$};\\
	0 & \text{otherwise}.
	\end{array}
	\right.
\end{equation*}
The key property is that $v_{b_v}(g_{(u,v)})$ is large enough so that in approximate equilibria, the buyer $b_v$ always wins the goods $g_{(u,v)}$ and the payment on this good is $\alpha_{b_u}$.

\paragraph{NOR gates.}

Fix a $\NOR$ gate $(u,v,w,\NOR)$ in $C$. 
We set the budget of $b_w$ to be $2.5$.

For the buyer $b_w$, we create an auxiliary buyer $c_w$ and create a good $g_w$ for them, namely, only $b_w$ and $c_w$ have non-zero values on the good $g_w$. In particular, we let 
\begin{equation*}
	v_{i}(g_{w})\coloneqq
	\left\{ \begin{array}{ll}
	18 & \text{if $i=b_w$};\\
	2 & \text{if $i=c_w$};\\
	0 & \text{otherwise}.
	\end{array}
	\right.
\end{equation*}
Again, we intend to make sure that $0.1\cdot v_{b_w}(g_{w})\leq (1-\sigma)v_{c_w}(g_{w})$, where $0.1$ is the intended lower bound for $\alpha_{b_w}$. We set the budget of $c_w$ to be $1000$.

We then construct two goods $g_{(u,w)}$ for $b_u$ and $b_w$ and $g_{(v,w)}$ for $b_v$ and $b_w$ respectively. The values of these goods to the buyers are as follows.

\begin{equation*}
	v_{i}(g_{(u,w)})\coloneqq
	\left\{ \begin{array}{ll}
	1 & \text{if $i=b_u$};\\
	1000 & \text{if $i=b_w$};\\
	0 & \text{otherwise}.
	\end{array}
	\right.
	\quad \text{and} \quad
	v_{i}(g_{(v,w)})\coloneqq
	\left\{ \begin{array}{ll}
	1 & \text{if $i=b_v$};\\
	1000 & \text{if $i=b_w$};\\
	0 & \text{otherwise}.
	\end{array}
	\right.
\end{equation*}
The key property is that $v_{b_w}(g_{(u,w)})$ and $v_{b_w}(g_{(v,w)})$ are large enough so that in approximate equilibria, the buyer $b_w$ always wins the goods $g_{(u,w)}$ and $g_{(v,w)}$ and the payment on these two goods is $\alpha_{b_u}+\alpha_{b_v}$.

\paragraph{NPURIFY gates.}
Fix a $\NPURIFY$ gate $(u,v,w,\NPURIFY)$. 
We set the budget of $b_v$ to be $1.8$, and set the budget of $b_w$ to be $2.8$.

For the buyer $b_v$, we create an auxiliary buyer $c_v$ and create a good $g_v$ for them, namely, only $b_v$ and $c_v$ have non-zero values on the good $g_v$. Similarly, for the buyer $b_w$, we create an auxiliary buyer $c_w$ and create a good $g_w$ for them, namely, only $b_w$ and $c_w$ have non-zero values on the good $g_w$. In particular, we let  
\begin{equation*}
	v_{i}(g_{v})\coloneqq
	\left\{ \begin{array}{ll}
	13.5 & \text{if $i=b_v$};\\
	1.5 & \text{if $i=c_v$};\\
	0 & \text{otherwise}.
	\end{array}
	\right.
	\quad \text{and} \quad
	v_{i}(g_{w})\coloneqq
	\left\{ \begin{array}{ll}
	18 & \text{if $i=b_w$};\\
	2 & \text{if $i=c_w$};\\
	0 & \text{otherwise}.
	\end{array}
	\right.
\end{equation*}
Similarly, we intend to make sure that $0.1\cdot v_{b_v}(g_{v})\leq (1-\sigma)v_{c_v}(g_{v})$ and $0.1\cdot v_{b_w}(g_{w})\leq (1-\sigma)v_{c_w}(g_{w})$. We set the budget of $c_v$ and $c_w$ to be $1000$.

We then construct two goods: $g_{(u,v)}$ for $b_u$ and $b_v$, and $g_{(u,w)}$ for $b_u$ and $b_w$ respectively. The values of these two goods to the buyers are as follows.

\begin{equation*}
	v_{i}(g_{(u,v)})\coloneqq
	\left\{ \begin{array}{ll}
	1 & \text{if $i=b_u$};\\
	1000 & \text{if $i=b_v$};\\
	0 & \text{otherwise}.
	\end{array}
	\right.
	\quad \text{and} \quad
	v_{i}(g_{(u,w)})\coloneqq
	\left\{ \begin{array}{ll}
	1 & \text{if $i=b_u$};\\
	1000 & \text{if $i=b_w$};\\
	0 & \text{otherwise}.
	\end{array}
	\right.
\end{equation*}
The key property is that $v_{b_v}(g_{(u,v)})$ and $v_{b_w}(g_{(u,w)})$ are large enough so that in approximate equilibria, the buyer $b_v$ always wins the goods $g_{(u,v)}$ and its payment on this good is $\alpha_{b_u}$; and the buyer $b_w$ always wins the goods $g_{(u,w)}$ and its payment on this good is $\alpha_{b_u}$.

\subsection{Proof of Correctness}
\begin{lemma}\label{lemma: general range ckk}
	Let $\alpha$ be the multipliers of any $(1/20,1/20,1/20)$-approximate pacing equilibrium of the constructed second-price pacing game. We have $\alpha_{c_v}\geq 1-\tau\geq 19/20$ and $\alpha_{b_v}\geq 0.1$ for all $v\in[n]$.
\end{lemma}
\begin{proof}
	We first note that in the construction of all gates, $c_v$'s budget is so large that by the no unnecessary pacing condition, we have $\alpha_{c_v}\geq 1-\tau$ for all $v\in[n]$. Now we show that $\alpha_{b_v}\geq 0.1$ for all $v\in[n]$. At a high level, the reasoning is exactly the same for the three different gates, namely, if $\alpha_{b_v}<0.1$, then $b_v$ would not be able to buy $g_v$, which (we will show) implies that $b_v$ will spend less than $19/20=(1-\gamma)$ fraction of their budget, leading to a contradiction.
	
	\textbf{$\NOT$ gates.} Assume for contradiction that $\alpha_{b_v}<0.1$. Then we know that $\alpha_{b_v}\cdot v_{b_v}(g_v)<(1-\sigma)\alpha_{c_v}\cdot v_{c_v}(g_v)$, which means that $b_v$ cannot spend any money on the good $g_v$. Then, $b_v$'s expense is at most $v_{b_u}(g_{(u,v)})=1$ (since it is a second-price auction). Note that the budget of $b_v$ is 1.5. Thus, $b_v$'s expense is at most $2/3<(1-\gamma)=19/20$ fraction of their budget, leading to a contradiction.
	
	\textbf{$\NOR$ gates.} Assume for contradiction that $\alpha_{b_w}<0.1$. Then we know that $\alpha_{b_w}\cdot v_{b_w}(g_w)<(1-\sigma)\alpha_{c_w}\cdot v_{c_w}(g_w)$, which means that $b_w$ cannot spend any money on the good $g_w$. Then, $b_w$'s expense is at most $v_{b_u}(g_{(u,w)})+v_{b_v}(g_{(v,w)})=2$ (since it is a second price auction). Note that the budget of $b_w$ is $2.5$. Thus, $b_w$'s expense is at most $4/5<(1-\gamma)=19/20$ fraction of their budget, leading to a contradiction.
	
	\textbf{$\NPURIFY$ gates.} 
	Assume for contradiction that $\alpha_{b_v}<0.1$. Then we know that $\alpha_{b_v}\cdot v_{b_v}(g_v)<(1-\sigma)\alpha_{c_v}\cdot v_{c_v}(g_v)$, which means that $b_v$ cannot spend any money on the good $g_v$. Then, $b_v$'s expense is at most $v_{b_u}(g_{(u,v)})=1$ (since it is a second-price auction). Note that the budget of $b_v$ is $1.8$. Thus, $b_v$'s expense is at most $10/18<(1-\gamma)=19/20$ fraction of their budget, leading to a contradiction.
	
	Assume for contradiction that $\alpha_{b_w}<0.1$. Then we know that $\alpha_{b_w}\cdot v_{b_w}(g_w)<(1-\sigma)\alpha_{c_w}\cdot v_{c_w}(g_w)$, which means that $b_w$ cannot spend any money on the good $g_w$. Then, $b_w$'s expense is at most $v_{b_u}(g_{(u,w)})=1$ (since it is a second-price auction). Note that the budget of $b_w$ is $2.5$. Thus, $b_w$'s expense is at most $2/5<(1-\gamma)=19/20$ fraction of their budget, leading to a contradiction.
	 
	 This finishes the proof.
\end{proof}
\begin{lemma}\label{lemma: potential winning goods ckk}
	Let $\alpha$ be the multipliers of any $(1/20,1/20,1/20)$-approximate pacing equilibrium of the constructed second-price pacing game. For any good $g_{(u,v)}$, the buyer $b_v$ buys the whole unit of $g_{(u,v)}$ and the expense on $g_{(u,v)}$ is $\alpha_{b_u}$.
\end{lemma}
\begin{proof}
	By \Cref{lemma: general range ckk}, we know that $\alpha_{b_v}\geq 0.1$. Then $\alpha_{b_v}\cdot v_{b_v}(g_{(u,v)})> v_{b_u}(g_{(u,v)})/(1-\sigma)\geq \alpha_{b_u}\cdot v_{b_u}(g_{(u,v)})/(1-\sigma).$ Since only buyers whose bids are close to the highest bid win the goods and pay the second price, we conclude that the buyer $b_v$ buys the whole unit of $g_{(u,v)}$ and the expense on $g_{(u,v)}$ is $\alpha_{b_u}$ (recall that in the construction, $v_{b_u}(g_{(u,v)})=1$).
\end{proof}

\begin{lemma}[Correctness of $\NOT$ gates]
	Let $\alpha$ be the multipliers of any $(1/20,1/20,1/20)$-approximate pacing equilibrium of the constructed second-price pacing game. For any $\NOT$ gate $(u,v)$, if $\alpha_{b_u}\leq 0.15$, then $\alpha_{b_v}\geq 0.9$; if $\alpha_{b_u}\geq 0.9$, then $\alpha_{b_v}\leq 0.15$.
\end{lemma}
\begin{proof}
	By \Cref{lemma: potential winning goods ckk}, we know that $b_v$ must buy the whole unit of $g_{(u,v)}$, may buy some amount of $g_v$, and does not buy any amount of other goods. So $b_v$'s expense is upper bounded by $\alpha_{b_u}+v_{c_v}(g_v)$. Furthermore, if $\alpha_{b_v}>0.15$, then $b_v$ also wins the whole unit of the good $g_v$, which implies that $b_v$'s expense equals $\alpha_{b_u}+v_{c_v}(g_v)$. 
	
	If $\alpha_{b_u}\leq 0.15$, then $b_v$'s expense is upper bounded by $0.15+1$. Recall that $b_v$'s budget is 1.5. So $b_v$ spends at most $1.15/1.5<19/20=(1-\gamma)$ fraction of their budget. Thus by the no unnecessary pacing condition, we have $\alpha_{b_v}\geq 1-\tau\geq 0.9$.
	
	If $\alpha_{b_u}\geq 0.9$, assume for contradiction that $\alpha_{b_v}>0.15$. Then $b_v$'s expense equals $\alpha_{b_u}+v_{c_v}(g_v)\geq 0.9+1>1.5$, violating the budget constrain and leading to a contradiction. Combining with \Cref{lemma: general range ckk} which shows $\alpha_{b_v}\geq 0.1$, we conclude that if $\alpha_{b_u}\geq 0.9$, then $\alpha_{b_v}\in[0.1,0.15]$.
\end{proof}

\begin{lemma}[Correctness of $\NOR$ gates]
	Let $\alpha$ be the multipliers of any $(1/20,1/20,1/20)$-approximate pacing equilibrium of the constructed second-price pacing game. For any $\NOR$ gate $(u,v,w)$, if $\alpha_{b_u}\leq 0.15$ and $\alpha_{b_v}\leq 0.15$, then $\alpha_{b_w}\geq 0.9$; if $\alpha_{b_u}\geq 0.9$ or $\alpha_{b_v}\geq 0.9$, then $\alpha_{b_w}\leq 0.15$.
\end{lemma}
\begin{proof}
	By \Cref{lemma: potential winning goods ckk}, we know that $b_w$ must buy the whole unit of $g_{(u,w)}$ and $g_{(v,w)}$, may buy some amount of $g_w$, and does not buy any amount of other goods. So $b_v$'s expense is upper bounded by $\alpha_{b_u}+\alpha_{b_v}+v_{c_w}(g_w)$. Furthermore, if $\alpha_{b_w}>0.15$, then $b_w$ also wins the whole unit of the good $g_w$, which implies that $b_w$'s expense equals $\alpha_{b_u}+\alpha_{b_v}+v_{c_w}(g_w)$.
	
	If $\alpha_{b_u}\leq 0.15$ and $\alpha_{b_v}\leq 0.15$, then $b_v$'s expense is upper bounded by $0.15+0.15+2$. Recall that $b_w$'s budget is 2.5. So $b_w$ spends at most $2.3/2.5<19/20=(1-\gamma)$ fraction of their budget. Thus by the no unnecessary pacing condition, we have $\alpha_{b_w}\geq 1-\tau=19/20>0.9$.
	
	If $\alpha_{b_u}\geq 0.9$ or $\alpha_{b_v}\geq 0.9$, assume for contradiction that $\alpha_{b_w}>0.15$. Then $b_w$'s expense equals $\alpha_{b_u}+\alpha_{b_v}+v_{c_w}(g_w)\geq 0.9+0.1+2>2.5$, violating the budget constrain and leading to a contradiction. Combining with \Cref{lemma: general range ckk} which shows $\alpha_{b_w}\geq 0.1$, we conclude that if $\alpha_{b_u}\geq 0.9$ or $\alpha_{b_v}\geq 0.9$, then $\alpha_{b_w}\in[0.1,0.15]$.
	
	This finishes the proof.
\end{proof}

\begin{lemma}[Correctness of $\NPURIFY$ gates]
	Let $\alpha$ be the multipliers of any $(1/20,1/20,1/20)$-approximate pacing equilibrium of the constructed second-price pacing game. For any $\NPURIFY$ gate $(u,v,w)$, we have $\alpha_{b_v}\leq 0.15$ or $\alpha_{b_w}\geq 0.9$. Furthermore, if $\alpha_{b_u}\leq 0.15$, then $\alpha_{b_v}\geq 0.9 $ and $\alpha_{b_w}\geq 0.9$; if $\alpha_{b_u}\geq 0.9$, then $\alpha_{b_v}\leq 0.15$ and $\alpha_{b_w}\leq 0.15$.
\end{lemma}
\begin{proof}
	By \Cref{lemma: potential winning goods ckk}, we know that $b_v$ must buy the whole unit of $g_{(u,v)}$, may buy some amount of $g_v$, and does not buy any amount of other goods. So $b_v$'s expense is upper bounded by $\alpha_{b_u}+v_{c_v}(g_v)$. Furthermore, if $\alpha_{b_v}>0.15$, then $b_v$ also wins the whole unit of the good $g_v$, which implies that $b_v$'s expense equals $\alpha_{b_u}+v_{c_v}(g_v)$.	
	
	Similarly, we know that $b_w$ must buy the whole unit of $g_{(u,w)}$, may buy some amount of $g_w$, and does not buy any amount of other goods. So $b_w$'s expense is upper bounded by $\alpha_{b_u}+v_{c_w}(g_w)$. Furthermore, if $\alpha_{b_w}>0.15$, then $b_w$ also wins the whole unit of the good $g_w$, which implies that $b_w$'s expense equals $\alpha_{b_u}+v_{c_w}(g_w)$.	
	
	We will show the following cases, given which the lemma directly follows.
	\begin{itemize}
		\item When $\alpha_{b_u}\leq 0.15$, we have $\alpha_{b_v}\geq 0.9$;
		\item When $\alpha_{b_u}>0.4$, we have $\alpha_{b_v}\leq 0.15$;
		\item When $\alpha_{b_u}\leq 0.4$, we have $\alpha_{b_w}\geq 0.9$;
		\item When $\alpha_{b_u}\geq 0.9$, we have $\alpha_{b_w}\leq 0.15$.
	\end{itemize}
	
	When $\alpha_{b_u}\leq 0.15$, $b_v$'s expense is upper bounded by $0.15+1$. Recall that $b_v$'s budget is $1.8$. So $b_v$ spends at most $1.15/1.8<19/20=(1-\gamma)$ fraction of their budget. Thus by the no unnecessary pacing condition, we have $\alpha_{b_v}\geq 0.9$.

	When $\alpha_{b_u}>0.4$, assume for contradiction that $\alpha_{b_v}>0.15$. Then $b_v$'s expense equals $\alpha_{b_u}+v_{c_v}(g_v)>0.4+1.5=1.9>1.8$, violating the budget constrain and leading to a contradiction. 
	
	When $\alpha_{b_u}\leq 0.4$, $b_w$'s expense is upper bounded by $0.4+2$. Recall that $b_w$'s budget is $2.8$. So $b_w$ spends at most $2.4/2.8<(1-\gamma)$ fraction of their budget. Thus by the no unnecessary pacing condition, we have $\alpha_{b_w}\geq 0.9$.
	
	When $\alpha_{b_u}\geq 0.9$, assume for contradiction that $\alpha_{b_w}>0.15$. Then $b_w$'s expense equals $\alpha_{b_u}+v_{c_w}(g_w)\geq 0.9+2>2.8$, violating the budget constrain and leading to a contradiction. 
	
	This finishes the proof.
\end{proof}
Given all the lemmas above, it is clear that given any pacing equilibrium $\alpha$, we can correctly extract a solution to the given $\purec$ instance. Furthermore, the construction of the auction instance can clearly be done in polynomial time. Thus we finished the proof of \Cref{thm: approx in CKK}.

\end{document}